\documentclass[a4paper,onecolumn,10pt]{article}
\pdfoutput=1
\usepackage[utf8]{inputenc}
\usepackage[english]{babel}
\usepackage[T1]{fontenc}
\usepackage{amsmath}
\usepackage{amsfonts}
\usepackage{amsthm}
\usepackage{hyperref}
\usepackage{graphicx}

\hypersetup{
    colorlinks=true,
    linkcolor=blue,
    citecolor=blue,
    urlcolor=blue
}

\DeclareMathOperator{\energy}{E}

\newtheorem{theorem}{Theorem}
\newtheorem{lemma}{Lemma}
\newtheorem{proposition}{Proposition}

\begin{document}

\title{Existence of abelian BPS vortices on surfaces with Neumann boundary conditions}

\date{}
\author{René García-Lara\thanks{Institute of Mathematics, 
National Autonomous University of Mexico.  
Cuernavaca, México. 
\texttt{Email}: rene.garcia@correo.uady.mx}}




\maketitle

\begin{abstract}
  Existence of abelian BPS vortices on a manifold with boundary
  satisfying Neumann boundary conditions is
  proved. Numeric solutions are constructed on the Euclidean disk, and the \(L^2\)-metric of the moduli space of one vortex located at the interior of a rotationally symmetry disk is studied. The results presented extend previous  work of Manton and Zhao on quotients of surfaces that admit a reflection.
\end{abstract}

\section{Introduction}\label{sec:motivation}

Vortices are well-known topological solitons in two dimensions. At critical coupling,
they have the interesting property that there is a
moduli space of vortices, which is also a K\"ahler manifold where geodesic curves
approximate the physically relevant dynamics of
fields on the surface~\cite{Stuart1994e}. Vortices of the Ginzburg-Landau functional
at critical coupling, or BPS vortices, satisfy a self-duality condition known as the
Bogomolny equations which on the other hand leads to an elliptic problem called the
Taubes equation. This problem is well studied for the Euclidean plane~\cite{taubes_arbitrary_1980} and for
closed surfaces with no boundary~\cite{bradlow_vortices_1990,garcia-prada_direct_1994,
  noguchi_yangmillshiggs_1987}.
For surfaces with boundary, Nasir proves the existence of vortices
whose Taubes equation satisfies Dirichlet boundary
conditions~\cite{nasir_study_1998}, his approach however cannot be applied
to the Taubes equation with Neumann boundary conditions.
Manton and Zhao~\cite{manton_neumann_2023_published} have recently found a Bradlow type
condition that guarantees the existence of vortices with Neumann boundary conditions in surfaces which are the quotient by a reflection of a smooth surface without boundary. They deduced that vortices can exist at the boundary of the surface,
where a vortex concentrates half the energy and magnetic flux of the respective
energy and magnetic flux of an interior vortex.
In this paper we study the Taubes equation to prove the existence of vortices
satisfying Neumann boundary conditions on an arbitrary Riemannian surface, thus
extending the result of Manton and Zhao to general surfaces. We also prove that
vortices at the boundary of the surface are half-vortices as in the previously studied case, in the sense that
they  have half the energy and magnetic flux of interior vortices. On a surface without  boundary, the moduli space of BPS-vortices has a well known Riemannian metric associated, together with a localization formula that relates geometric properties of moduli space to purely local data corresponding to the relative position of the vortices. We exhibit a family of surfaces where, although there is a well defined metric on moduli space, it is no longer true that the metric depends only on vortex position, in contrast to the well known case of vortices on a closed Riemannian surface and of vortices with Dirichlet boundary conditions~\cite{samolsMathematicalPhysicsVortex1992a,nasir_study_1998}.   

The paper is
organised as follows.
In Section~\ref{sec:bps-vortices-neumann} we describe the BPS vortex model of the Ginzburg-Landau functional and state the
main results,~theorems~\ref{thm:existence-vortices}  and~\ref{thm:tabues-problem}.
In Section~\ref{sec:taubes-reg} we
use techniques of elliptic partial differential equations to prove that, up to
gauge equivalence, there exists exactly one solution of the field equations. With
standard methods of elliptic operators we prove that this solution is smooth. In Section~\ref{sec:numerics} we present a pair of vortex solutions computed numerically on the Euclidean disk. In Section~\ref{sec:l2-metric} we introduce the \(L^2\) metric of moduli space and provide an example where the properties of the metric differ from the classical results on a Riemann surface without boundary. We finalise with some comments.

\section{Statement of main results}\label{sec:bps-vortices-neumann}

Let \( S \) be the interior of a smooth, orientable and oriented compact surface
\( \bar S \) with non empty boundary
\( \partial S \) equipped with a Riemannian metric and an Hermitian \( U(1) \) bundle,
such that for any pair \( (\phi, A) \) of a section \( \phi \) and a \( U(1) \)-connection \( A \),
the static Ginzburg-Landau energy is
\begin{align}
  \label{eq:gl-energy}
  \energy(\phi, A) = \frac{1}{2} \int_{S} |F|^2 + |D\phi|^2 +
  \frac{1}{4} {(1 - |\phi|^2)}^2 dV,
\end{align}
where \(F = dA\) is the curvature form of the connection and \(D\phi \)
is the associated
covariant derivative, represented in a local trivialisation \((x^1, x^2)\) as
\((\partial_j\phi - iA_j\phi)\,dx^j\).
We are concerned with BPS vortices of the Ginzburg-Landau
functional, these are energy minimizing pairs \((\phi, A)\) which are
solutions of the Bogomolny equations
\begin{align}
  \label{eq:bog-1}
  \bar\partial_A \phi & = 0,                                     \\
  \label{eq:bog-2}
  B                   & = \frac{1}{2} \left(1 - |\phi|^2\right),
\end{align}
where \(B = *F\) is the magnetic field across the surface and
 \(\bar{\partial}_A\phi \) is the anti-holomorphic partial derivative of \( \phi \),
which in a local coordinate system is the
section \( D_1\phi + i\,D_2\phi \). The current of a pair \( (\phi,
  A) \) is the vector field
\begin{align}
  J = \Re (i\bar\phi D\phi). \label{eq:current}
\end{align}
Let \( t \) and \( n \) be the tangent and outward pointing normal vector
fields on \( \partial S \), we choose the orientation of the boundary so that
\( (n,t) \) is a positively oriented
basis of \( T\bar S|_{\partial S} \). Let \( J_t \) and \( J_n \) be the tangent and normal projections of \( J \) at the boundary,
we focus on the condition \( J_t = 0 \) which we will call the
Neumann boundary condition from now onwards. We can decompose the Higgs field \( \phi \)
as \( \phi = e^{h/2 + i\chi} \),
where \( h = \log |\phi|^2 \) is a function with logarithmic singularities
and \( \chi \) is an unwrap of the phase, which need not be
unique; however, any two such phases will differ by an integer
multiple of \( 2\pi \). A direct calculation shows that by Equation~\eqref{eq:bog-1}, the Neumann boundary condition is
equivalent to \( \partial_n h = 0 \), justifying the name. 
Manton and Zhao prove~\cite{manton_neumann_2023_published} that if the zeroes of
\( \phi \) are located in the interior of
\( S \), and if \( (\phi, A) \) is a solution of the Bogomolny
equations, then the total energy of the fields is quantised,
\begin{equation}
  \label{eq:quant-energy}
  \energy(\phi,A) = n\pi,
\end{equation}
where \( n \) is the total number of zeroes of \( \phi \) counted with
multiplicity. Furthermore, the total magnetic flux is also quantised:
\begin{equation}
  \label{eq:quant-flux}
  \int_{\bar S} F = 2n\pi,
\end{equation}
in comparison, Nasir showed that for BPS vortices with Dirichlet boundary
conditions, only the energy remains quantised as in the case of closed
surfaces. 
For surfaces which are the quotient by a reflection of a closed surface,
Manton and Zhao proved that there are BPS vortices with zeroes on the boundary of 
the surface, which is not possible for Dirichlet boundary conditions. Let us call a set of points with multiplicities to any set  \( M \subset \bar S \times \mathbb{N} \), moreover, if \( M =
  {\{(s_i,m_i)\}}_{i=1,\ldots,n} \) is one such set, we define the total
multiplicity of \( M \) as the number \( \sum_{i=1}^n m_i \). The main result of the paper is the following theorem.

\begin{theorem}\label{thm:existence-vortices}
  Let \( S \) be the interior of a compact Riemannian surface with
  boundary \( \partial S \) and area \( A \). Given two finite sets of points with multiplicities,
  \begin{equation*}
    \mathcal{N} \subset S \times \mathbb{N} \qquad \text{and} \qquad
    \mathcal{M} \subset \partial S \times \mathbb{N},
  \end{equation*}
  \( \mathcal{N} \) or \( \mathcal{M} \) possibly empty but \( \mathcal{N} \cup
    \mathcal{M} \neq \emptyset \),
  with total multiplicities \( N \) and \( M \) respectively, then the following
  statements are equivalent:
  \begin{enumerate}
    \item Up to gauge
          equivalence, there exists exactly
          one solution \( (\phi, A) \) of the Bogomolny
          equations~\eqref{eq:bog-1}-\eqref{eq:bog-2} with Neumann boundary
          conditions, such that \( \mathcal{N}\cup \mathcal{M} \) is the set of
          zeroes of \( \phi \) counted with multiplicity.
    \item The condition
          \begin{equation}
            N + \frac{M}{2} < \frac{A}{4\pi}\label{eq:bradlow-cond}
          \end{equation}
          holds.
  \end{enumerate}
\end{theorem}

Condition~\eqref{eq:bradlow-cond} is a Bradlow type condition,
previously found in~\cite{manton_neumann_2023_published} for surfaces which
are quotients with respect to a Riemannian reflection.
We will deduce this condition from the analysis of
the related elliptic problem that we describe in the next section, where
we prove that Theorem~\ref{thm:existence-vortices} holds, provided the elliptic problem has a unique solution.

\subsection{The Taubes equation with Neumann boundary conditions}

If \( (\phi, A) \) is a solution of the Bogomolny equations, let us define the 
function \( h = \log |\phi|^2 \).
Let \( \{X_1,\ldots, X_N\}\subset S \) be a set of interior points
of the surface, possibly repeated, and let \( \{W_1,\ldots, W_m\}
  \subset \partial S \) be a set of boundary points also counted with multiplicity.
Proving existence of a solution of the Bogomolny equations
reduces to prove the existence of a solution of the elliptic problem
\begin{align}
  \nabla^2 h                  & = e^h - 1 + 4\pi \sum_{k} \delta^2_{X_k},\label{eq:taubes} \\
  -\partial_n h|_{\partial S} & = 2\pi \sum_{j} \delta^1_{W_j},\label{eq:boundary-cond}
\end{align}

where \( \nabla^2 \) is the Laplace-Beltrami operator and \( \partial_{n}h \) is
the directional derivative at the boundary in the outward pointing
direction. Equation~\eqref{eq:taubes} is called the Taubes
equation, for this equation, existence of solutions in the plane and on a
closed surface is well
known~\cite{taubes_arbitrary_1980,garcia-prada_direct_1994}. 

\begin{theorem}\label{thm:tabues-problem}
  Let \( S \) be a smooth Riemannian surface with boundary \( \partial S \) and
  let \( \{ X_1, \ldots, X_N \} \subset S \), \( \{
    W_1,\ldots, W_M\} \subset \partial S \) be sets of points
  counted with repetition, then there exists exactly one solution
  of the Taubes equation~\eqref{eq:taubes} with boundary condition~\eqref{eq:boundary-cond},
  if and only if the Bradlow condition~\eqref{eq:bradlow-cond} holds.
  This solution is smooth away of the set of core points given.
\end{theorem}

\noindent
Assume for a moment that Theorem~\ref{thm:tabues-problem} is proved, then the main theorem is a consequence of this as we prove in the following paragraphs. 

\begin{proof}[Proof of Theorem~\ref{thm:existence-vortices}]
  By equation~\eqref{eq:bog-2}, we can find the curvature 2-form
  \(F =*B \), hence, we can recover the connection \( A \) up to a gauge function.
  For \(\chi \) we follow a known prescription: For any singular point \(X_k \) of \( h \), in accordance to~\cite{taubes_arbitrary_1980} we choose a coordinate system
  \((U,\varphi) \), such that \(\varphi(X_k) = 0 \) 
  and define locally \(\chi_k(X) = n_k\arg(\varphi(X) - \varphi(X_k)) \), where \(n_k \) is the multiplicity of \( X_k \) and \( X \in U \setminus \{X_k\} \). Extending by a
  suitable family of bump functions, we can obtain a globally defined function
  \(\chi \) with the same singularity type as \(h/2 \). If we define \(\phi =
    e^{h/2 + i\chi} \), each singularity of \( h \) corresponds to a zero of \( \phi \) with
  prescribed multiplicity. Moreover, the first Bogomolny
  equation~\eqref{eq:bog-1} is equivalent to
  \begin{equation}
    \label{eq:bog-1-polar}
    A = -\frac{1}{2} *dh + d\chi.
  \end{equation}
  Equation~\eqref{eq:bog-1-polar} determines a connection \( A \) away of the
  singularities of \( h \) and extends smoothly to the singular set.
  By equation~\eqref{eq:bog-2}, away of the singular points,
  \begin{equation*}
    *dA = -\frac{1}{2}*d*dh = - \frac{1}{2}\nabla^2h = \frac{1}{2}(1 - e^h) =
    \frac{1}{2}(1 - |\phi|^2).
  \end{equation*}
  Thus \((\phi, A) \) is the solution of the Bogomolny equations. Conversely, for
  any solution of the Bogomolny equations, there is a decomposition of \( \phi \) in
  a pair \((h, \chi) \), such that \(\chi \) is locally well defined up to a gauge 
  function, and \((h, \chi) \) are solutions of equations~\eqref{eq:bog-1-polar},~\eqref{eq:taubes},~\eqref{eq:boundary-cond}.
  Therefore, up to gauge equivalence there exists exactly one solution to the
  Bogomolny equations.
\end{proof}
Equation~\eqref{eq:taubes} implies \( h \) is a function of logarithmic divergence at the singularities, with this observation, we can compute the total magnetic flux for arbitrary surfaces, extending the previous result for quotients of surfaces with reflection symmetry. 
\begin{proposition}
  If \((\phi, A) \) is a solution of the Bogomolny equations, with a total of \(N\) 
  zeroes of \( \phi \) in the interior of the surface and \(M \) zeroes at the boundary,
  then the total magnetic flux is
  \begin{equation}
    \label{eq:total-flux}
    \Phi = (2N + M)\pi.
  \end{equation}
\end{proposition}

\begin{proof}
  Let \(X_k \in S \) be a zero of \( \phi \) and let \(\varphi: U_k \to
    \mathbb{C} \) be a normal coordinate system such that \(\varphi(X_k) = 0 \). Let \( r_k(X) = |\varphi(X)| \) for \( X \in U_k \). Analogously, if \( W_j \in \partial S \) is the
  position of a zero of \( \phi \), let \( \varphi: U_j \to H^+ \) be a normal coordinate
  chart such that \( \varphi(W_j) = 0 \), where \( H^+ \) is the upper semi plane of complex numbers \( z \) such that \( \Im(z) \geq 0 \). For any \( W \in U_j \), we define \( r_j(W) = |\varphi(W)| \). Let \( \epsilon > 0 \) be a positive number smaller than the 
  injectivity radius of the metric, let \( B_{\epsilon}(X_k) \) be the geodesic
  ball around \( X_k \) of radius \( \epsilon \) and  let \( B_{\epsilon}(W_j) \) be the geodesic semi ball around \( W_j \) of the same radius. If \(U_\epsilon = \bar S \setminus (\cup_k B_{\epsilon}(X_k) \cup_j B_{\epsilon}(W_j)) \),
  the total magnetic flux can be computed as,
  \begin{multline}
    \Phi = \lim_{\epsilon\to 0}\int_{U_\epsilon} F
    = - \frac{1}{2}\lim_{\epsilon\to 0} \int_{U_\epsilon} d*dh
    = - \frac{1}{2} \lim_{\epsilon\to 0} \int_{\partial U_{\epsilon}}
    *dh \\
    = \frac{1}{2} \sum_k \lim_{\epsilon\to 0} \int_{\partial B_{\epsilon}(X_k)}
    *dh + \frac{1}{2} \sum_j \lim_{\epsilon\to 0} \int_{\partial B_{\epsilon}(W_j)}
    *dh,
  \end{multline}
  where in the second equality we used equation~\eqref{eq:bog-1-polar} and in the last equation, \( \partial B_{\epsilon}(W_j) \) is the semi-sphere 
  \( \{W \in U_j \mid |r_j(W)| = \epsilon \} \). We also note that each boundary term in the last equality is oriented in the orientation induced by the
  outward-pointing normal with respect to the respective centre. Since \( h \) is the solution of equation~\eqref{eq:taubes}, there are locally defined smooth functions \( h_k:U_k \to \mathbb{R} \) such that \(h|_{U_k} = 2n_k\log(r_k) + h_k \),
  where \(n_k \) is the multiplicity of \(X_k \), similarly, there are well-defined smooth functions \(h_j: U_j \to \mathbb{R} \) such that \(h|_{U_j} = 2m_j\log r_j + h_j \), whence,
  \begin{align}
    \frac{1}{2}\lim_{\epsilon\to 0} \int_{\partial B_{\epsilon}(X_k)} *dh & =
    n_k\lim_{\epsilon\to 0} \int_{\partial B_{\epsilon}(X_k)} *d\log(r_k)  \nonumber                                                   \\
                                                                          & = n_k\lim_{\epsilon\to 0} \frac{1}{\epsilon}\int_{\partial
    B_{\epsilon}(X_k)} *dr_k \nonumber                                                                                                 \\
                                                                          & = 2n_k\pi,
  \end{align}
  since \(*dr_k = \epsilon d\theta_k \) where \(\theta_k(X) = \arg(\varphi(X)) \) for \(X \in U_k\setminus \{X_k\} \). Finally, for zeroes at the boundary, the computation is analogous, except that the domain of the argument is a semicircle, hence the integral is half of the previously computed value, and the proposition follows.
\end{proof}

The classical Bogomolny trick relates the total energy of vortices to the magnetic flux by means of the relation,
\begin{align}
  \operatorname{E} = \frac{1}{2} \Phi + \frac{1}{2} \int_{\bar{S}} \left(
  \Im \left(\langle
    D\phi, *D\phi
    \rangle\right) - |\phi|^2 F \label{eq:bog-trick-integral}
  \right),
\end{align}
where the bracket \(\langle \cdot, \cdot \rangle \) denotes the complex inner product on the tangent bundle. Manton and Zhao deduce that in the interior of  any compact surface, the integrand in Equation~\eqref{eq:bog-trick-integral} is \( d J \), where \( J \) is the current defined on Equation~\eqref{eq:current}.
Recall the Neumann condition states \(J_{t} = 0\), where \( J_{t} \) is the tangential component of the current at the boundary. This implies that the energy is
\begin{align}
  \operatorname{E} & = \frac{1}{2} \Phi + \frac{1}{2} \int_{\bar{S}} dJ \nonumber        \\
                   & = \frac{1}{2} \Phi + \frac{1}{2} \int_{\partial{S}} J_{t} \nonumber \\
                   & = \frac{1}{2} \Phi, \label{eq:flux-integral}
\end{align}
whence, for any surface, not necessarily of reflection type, the Energy of a vortex system is also given by Equation~\eqref{eq:quant-energy} whenever we weight each vortex at the boundary as half a vortex located in the interior.
In the next section we will prove Theorem~\ref{thm:tabues-problem}, the plan
for the proof is the following. Firstly, assume the existence
of a pair of Green functions \(G \) and \( H \), such that for \( Z \in \bar S \),

\begin{align}
  \nabla^2 G(\cdot, Z)                 & = \delta_Z, \\
  \partial_n G(\cdot, Z)|_{\partial S} & = 0,
\end{align}

and

\begin{align}
  \nabla^2 H(\cdot, Z)                 & = \frac{1}{A}, \\
  \partial_n G(\cdot, Z)|_{\partial S} & = 0,
\end{align}

where \( A \) is the area of \( S \) and \( \partial_n G \) is the normal exterior
derivative of \( G(\cdot, Z) \) at the boundary. We use these functions to define
the function \( \tilde h: \bar{S} \to \mathbb{R} \) such that
\begin{align}
  h = \tilde h + 4\pi \sum_k G(\cdot, X_k) + 2\pi \sum_j H(\cdot, W_j).
\end{align}
With this setup, \( h \) is a solution of the Taubes equation if and only if
\( \tilde h \) is a solution of the
elliptic problem
\begin{align}
  \nabla^2 \tilde h   & = e^{\tilde h + u_0} - 1 - \left(N + \frac{M}{2}\right) \frac{4\pi}{A},\label{eq:taubes-reg} \\
  \partial_n \tilde h & = 0,
\end{align}
where
\begin{align}
  u_0 = - 4\pi \sum_k G(\cdot, X_k) - 2\pi \sum_j H(\cdot, W_j).\label{eq:u0}
\end{align}
We call equation~\eqref{eq:taubes-reg} the regularised Taubes
equation, this is a Kazdan-Warner type equation~\cite{kazdan1974curvature}, with
the difference that the surface has a boundary. In the next section,
we will rely on the results
of the reference~\cite{noguchi_yangmillshiggs_1987} to justify
the existence of a unique solution of the regularised Taubes equation
and that this solution is smooth, this claim will imply
Theorem~\ref{thm:tabues-problem}.

\section{The regular elliptic problem}\label{sec:taubes-reg}

In this section we aim to prove the existence of a unique solution of the
regularised Taubes
equation~\eqref{eq:taubes-reg} with Neumann boundary conditions. It is simpler to consider the more general equation
\begin{align}
  \nabla^2 v = pe^{2v} + K,\label{eq:noguchi-problem}
\end{align}
where \(p: \bar S \to \mathbb{R} \) is a smooth function, positive except for a
finite number of zeroes and \(K \) is a function satisfying the condition
\begin{align}
  \int_{\bar S} K dV < 0.\label{eq:noguchi-condition}
\end{align}
For compact surfaces without boundary, there exists a unique solution
of~\eqref{eq:noguchi-problem} with the condition given on \(K \)~\cite[cf. Sec.
  IV]{noguchi_yangmillshiggs_1987}, we will extend the techniques of the
reference to surfaces with boundary, to this end, let \( H^1 \) be the
Sobolev space \( W^{1,2}(\bar S) \) of square integrable functions on \( \bar S \) with weak square integrable first derivatives. In the
sequel we denote by \( \bar u = A^{-1}\,\int_{\bar S}\,u\,dV \) the average of an integrable function \( u \), where
\( A \) is the surface area.
We establish the following lemmas.
\begin{lemma}\label{lem:exp-u-bound}
  If \( ||\cdot|| \) is the \( L^2 \) norm in \( \bar S \) and \( u \in H^1 \), then the function \( e^u \) is integrable and there exist constants \( \beta, \gamma \) such that for any \( \alpha  > 0 \),
  \begin{equation}\label{eq:trudinger}
    \int_{\bar S} e^{\alpha |u|} dV \leq \gamma \exp\left(\alpha |\bar u| +
    \frac{{(\alpha ||\nabla u||)}^2}{4\beta}\right).
  \end{equation}
\end{lemma}

\begin{lemma}\label{lem:exp-u-convergence}
  If \(u_j \to u \) weakly in \( H^1 \), then \( e^{u_j} \to e^u \) strongly in \( L^2 \).
\end{lemma}
For surfaces without boundary, the proof of both lemmas can be found
in~\cite{kazdan1974curvature}. The first lemma describes a Trudinger-Moser
type inequality, the proof relies on the basic Sobolev identity
\begin{equation*}
  ||u||_p \leq C p^{1/2} ||\nabla u||_{2},\qquad p \geq 1,
\end{equation*}
and on the Poincaré inequality, which are both valid for surfaces with
boundary, whence~\eqref{eq:trudinger} remains valid in this case. The
second lemma is a consequence of the first one, the proof remains unchanged for
surfaces with boundary.

If \( (\cdot.\cdot) \) is the inner product on \( W^{1,2}(\bar S) \), by the Riesz representation lemma, there are unique maps \( L, P: H^1 \to H^1 \) such that, for any pair \( u, v \in H^1 \),
\begin{align*}
  (Lu, v) = \int_{\bar S} \langle \nabla u, \nabla v\rangle dV,\qquad
  (P(u),v) = \int_{\bar S}\,pe^{2u}v dV.
\end{align*}
We note that the functional \( P \) is well defined since as a consequence of
Lemma~\ref{lem:exp-u-bound}, \( e^u \in L^p \) for any \( p > 0 \), hence, for \( u, v \in H^1 \) we have,
\begin{align*}
  \left| \int_{\bar S} pe^{2u}v dV \right| & \leq \left(\max_{\bar S} p\right)
  ||e^{2u}||\; ||v|| < \infty,
\end{align*}
where we used the fact that \( p \) is smooth and the Cauchy-Schwarz inequality.
We also note that the kernel of \( L \) is the set of constant functions, which is
isomorphic to \( \mathbb{R} \), moreover, the orthogonal complement of \( \ker L \) in
\( H^1 \) is the subspace,
\begin{align*}
  \mathcal{X} = \{u \in H^1\mid \bar u = 0\}.
\end{align*}
Any \( u\in H^1 \) decomposes uniquely as \( u = \bar u + w \) for some \( w \in
  \mathcal{X} \), our task is to show
that~\eqref{eq:noguchi-problem} has a weak solution,
in the sense that there exists a function \( u \in H^1 \) such that for
all \( v \in H^1 \), the equation
\begin{align}
  \label{eq:weak-noguchi-problem}
  \int_{\bar S} \langle \nabla u,\,\nabla v\rangle \, dV +
  \int_{\bar S}\,pe^{2u}\,v\,dV = -\int_{\bar S}\,K\,v\,dV,
\end{align}
holds.
Assume for a moment that this is the case, if \( u = c + w \), \( c \in \mathbb{R} \),
\( w\in \mathcal{X} \), then substituting \( v = 1 \) in
equation~\eqref{eq:weak-noguchi-problem}, we find
\begin{align}
  \label{eq:c-cond}
  e^{2c}\,\int_{\bar S}\,pe^{2w}\,dV = -\int_{\bar S}\,K\,dV.
\end{align}
Hence,
\begin{align}
  \label{eq:noguchi-c}
  c = \frac{1}{2}\left(\log\left(-\int_{\bar S}\,K\,dV\right) -
  \log\left(\int_{\bar S}\,pe^{2w}\,dV\right)\right),
\end{align}
i.e.\  \( c \) is uniquely determined by \( w \), whence~\eqref{eq:noguchi-c}
determines a function \( c: \mathcal{X} \to \mathbb{R} \).
Since \( H^1 = \mathbb{R} \oplus \mathcal{X} \), the projection operator \( \pi^2: H^1 \to \mathcal{X} \)
determines a functional \( T: \mathcal{X} \to \mathcal{X} \) defined as,
\begin{equation}
  T(w) = L w + \pi^2P(c(w) + w).\label{eq:op-T}
\end{equation}

\begin{lemma}\label{lem:noguchi-steps}
  \( T \) is an homeomorphism of \( \mathcal{X} \) onto \( \mathcal{X} \).
\end{lemma}
\begin{proof}
  The proof can be found in~\cite[Sec.~IV]{noguchi_yangmillshiggs_1987} for
  closed surfaces, however, it also applies for surfaces with boundary
  once we have established that
  lemmas~\ref{lem:exp-u-bound} and~\ref{lem:exp-u-convergence} hold.
\end{proof}

\begin{proposition}
  There exists a unique solution of~\eqref{eq:noguchi-problem} with Neumann
  boundary conditions if and only if condition~\eqref{eq:noguchi-condition}
  holds. This solution is of class \( C^{\infty} \).
\end{proposition}
\begin{proof}
  Let \( k \in H^1 \) be the only function such that for all \( v \in H^1 \),
  \begin{equation*}
    (k,v) = -\int_{\bar S} Kv dV.
  \end{equation*}
  By Lemma~\ref{lem:noguchi-steps} there exists a unique \( w \in \mathcal{X} \) such
  that for any \( v\in \mathcal{X} \),
  \begin{equation*}
    Lw + \pi^2P(c(w) + w) = \pi^2k.
  \end{equation*}
  Let \( u = c(w) + w \), then for any \( v \in \mathcal{X} \).
  \begin{equation*}
    \int_{\bar S} \langle \nabla u,\,\nabla v\rangle dV +
    \int_{\bar S}\,\left(p\exp(2u) - \overline{p\exp({2u})}\right)v dV
    = -\int_{\bar S}\,\left(K - \overline K\right)\,v\,dV,
  \end{equation*}
  since the average of any \( v \in \mathcal{X} \) is 0, we deduce,
  \begin{align}
    \label{eq:weak-noguchi-problem-X}
    \int_{\bar S}\langle \nabla u,\,\nabla v\rangle \, dV +
    \int_{\bar S}\,pe^{2u}v dV = -\int_{\bar S}\,KvdV.
  \end{align}
  As a consequence of~\eqref{eq:c-cond},
  Equation~\eqref{eq:weak-noguchi-problem-X} is
  valid for all \( v \in H^1 \). Therefore, \( u \) is a weak solution
  of~\eqref{eq:noguchi-problem}. By the usual elliptic estimates, \( u \) is a
  function in the Sobolev space \(W^{2,2} \), whence, it makes sense to use the
  divergence theorem in order to convert
  Equation~\eqref{eq:weak-noguchi-problem-X} into
  \begin{align}
    \label{eq:weak-noguchi-neumann}
    -\int_{\bar S} \nabla^2 u\,vdV +
    \int_{\bar S}pe^{2u}v dV + \int_{\partial S}\partial_{n} u\,v ds =
    -\int_{\bar S}\,Kv dV.
  \end{align}
  This equation is valid for any \( v \in H^1 \), in particular, for all
  \( v\in H^1 \) such that \( v|_{\partial S} = 0 \), whence \( u \) is a solution
  of~\eqref{eq:noguchi-problem} in the interior of the surface,
  which implies
  \begin{align}\label{eq:weak-laplacian}
    -\int_{\bar S} \nabla^2 u\,vdV +
    \int_{\bar S}\,pe^{2u}v dV = -\int_{\bar S}Kv dV
  \end{align}
  for all \( v \in H^1 \). Equations~\eqref{eq:weak-laplacian}
  and~\eqref{eq:weak-noguchi-neumann} imply \( \partial_n u = 0 \) and thus
  \( u \) is the unique solution of the Neumann problem.
\end{proof}

\begin{proof}[Proof of Theorem~\ref{thm:tabues-problem}]
  The function \(u_0 \) defined in~\eqref{eq:u0} is smooth, except for a finite
  set of singularities, moreover, we know that for any \(\epsilon > 0 \) smaller
  than the injectivity radius of the metric and for any singular point \(P \),
  there is a  smooth function \(\tilde u: B_{\epsilon}(P) \to \mathbb{R} \) such
  that \(u_0 = n\log r + \tilde u \), where \(n \) is a positive integer and
  \(r(X) \) is the Riemannian distance of the point \(X \) to \(P \). Within the
  injectivity radius, \(r \) is a smooth function, hence the function \(p = e^{u_0} \)
  is smooth in all the surface, moreover, if we define the constant function
  \begin{align}\label{eq:contant-k}
    K = - 1 + \left(N + \frac{M}{2}\right) \frac{4\pi}{A},
  \end{align}
  then by the Bradlow bound, condition~\eqref{eq:noguchi-condition} is
  fulfilled. Therefore, there exists a unique solution \(\tilde h \) of the
  regularised Taubes equation such that \(\partial_n \tilde h = 0 \) and this
  function is also smooth. Defining \(h = u_0 + \tilde h \) concludes the proof of the theorem.
\end{proof}

\section{Numerical results}\label{sec:numerics}

We solved numerically the Taubes equation on an Euclidean disk of radius 3, so that the Bradlow bound is satisfied for a single vortex, located either at the origin or at the boundary of the disk. In the first case, we defined \(\tilde{h} = h - \log |z|^2 \), then Taubes equation is equivalent to,
\begin{align}
  \nabla^2 \tilde{h} &= |z|^2 e^{\tilde{h}} - 1, \\
  \partial_{n} \tilde{h} &= -\frac{2}{3}.
\end{align}
Since the problem is rotationally symmetric, \(\tilde{h} \) is a function of the distance to the origin, in this case, \(\tilde{h} \) is the solution of the following boundary value problem, 
\begin{align}
\tilde{h}'' + \frac{1}{r}\tilde{h}' &= r^2 e^{\tilde{h}} - 1, & r \in (0, 3), \\
\tilde{h}'(0) &= 0, \\
\tilde{h}'(3) &= -2/3.
\end{align}
To avoid the singularity at the origin, we solved by a shooting method, expanding \(\tilde{h} \) in Taylor series around the origin, we found that for the differential equation to hold up to third order, \(\tilde{h} \) must be of the form, 
\begin{equation}\label{eq:taylor-shooting-parameter}
  \tilde{h} = h_{0} - \frac{1}{4} r^2 + \frac{e^{h_{0}}}{16} r^4 + \mathcal{O} (r^5),
\end{equation}
where \(h_{0} \) is the shooting parameter. We took an initial point \( \epsilon = 10^{-8} \) and solved the differential equation in the interval 
\([\epsilon, 3] \) with initial conditions \((\tilde{h}(\epsilon), \tilde{h}'(\epsilon))\) approximated by the Taylor expansion~\eqref{eq:taylor-shooting-parameter} until we reached the boundary condition at \(r = 3 \) within a tolerance of \(10^{-6}\). Figure~\ref{fig:1-vortex-disk} shows the profile of \(|\phi|^2 \), and the energy and magnetic field distributions of a rotationally symmetric single vortex located at the centre of the Euclidean disk of radius 3. We also solved for a vortex located at the boundary point \(z = 3\), in this case we defined \(\tilde{h} = h - \log |z - 3|^2\), then \(h \) is a solution of Taubes equation if and only if \(\tilde{h} \) solves the following problem on the disk,
\begin{align}
  \nabla^2 \tilde{h} &= |z - 3|^2 e^{\tilde{h}} - 1, \\
  \partial_{n} \tilde{h} &= -\frac{1}{3}.
\end{align}
In this case, the elliptic problem was solved by a finite element method, using the Fenix library, Figure~\ref{fig:2-half-vortex-disk} shows the energy, magnetic flux and \(|\phi|^2 \) for this configuration. In both cases, it is evident that the energy of the vortex concentrates near the vortex position, in the first case, the energy distribution is symmetric respect to the vortex position, on the other hand, in the second case, it is interesting to notice that the maximum energy is not located at the boundary of the Euclidean disk.

\begin{figure}[h]
  \centering
  \includegraphics[width=0.75\textwidth]{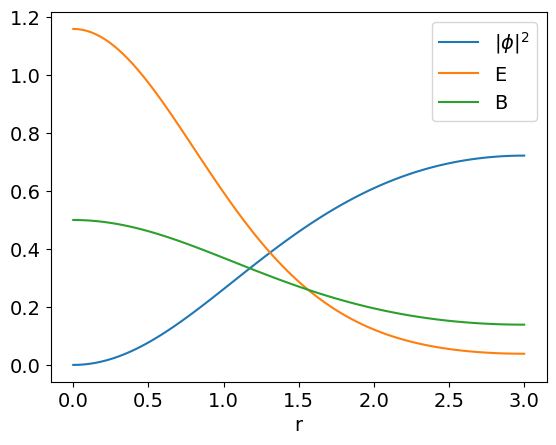}
  \caption{Energy, magnetic flux and Higgs field modulus of a single vortex located at the center of a disk with Neumann boundary conditions.}\label{fig:1-vortex-disk}
\end{figure}
\begin{figure}[h]
  \centering
  \includegraphics[width=\textwidth]{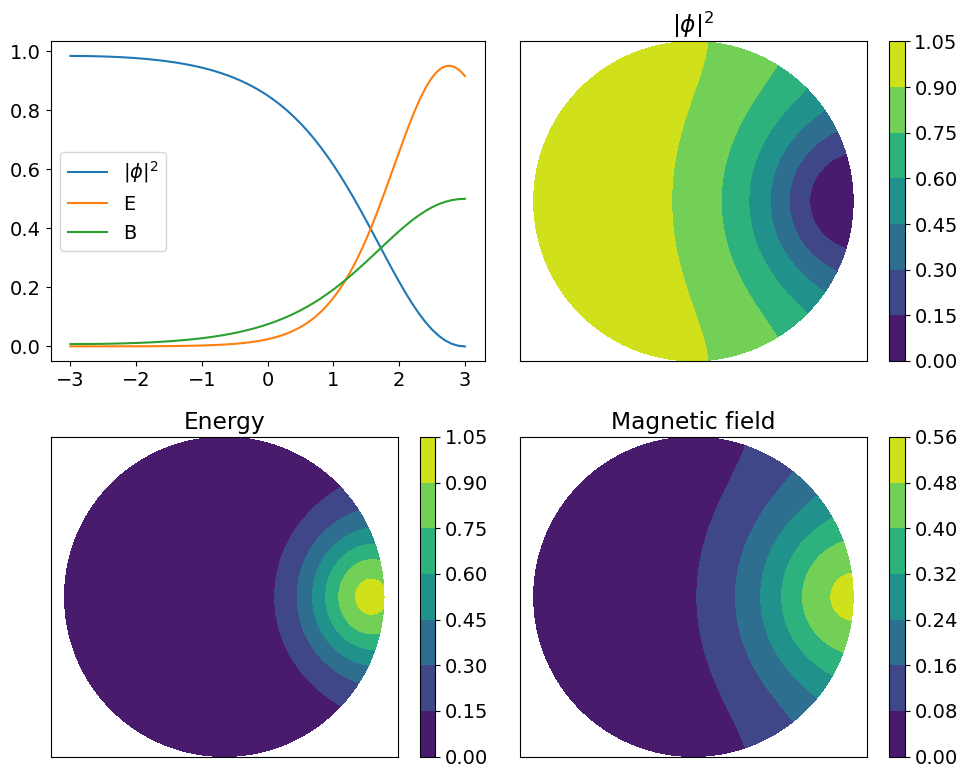}
  \caption{Energy, magnetic flux and Higgs field modulus of a half-vortex located at the boundary of a radius 3 disk with Neumann boundary conditions.}\label{fig:2-half-vortex-disk}
\end{figure}

\section{Comparing the geometry of the moduli space of one vortex}\label{sec:l2-metric}

For vortices on a Riemann surface and vortices with Dirichlet boundary conditions, the moduli space comes endowed with a K\"ahler metric, named the \(L^2\) metric, and a formula that relates the metric to local data depending only on relative vortex positions. For Neumann boundary conditions, we will exhibit an example where this is no longer the case, in fact, there is a family of surfaces where the metric of the moduli space of a single vortex depends on boundary data. Let \(X \in S \) be the position of a vortex in the interior of a Riemann surface, we denote by \(h(x; X)\) the solution of Taubes equation corresponding to this vortex, if \(F: S \to S\) is an isometry, uniqueness of the solution \(h\) implies \(h(F(x); F(X)) = h(x; X)\). Before providing the example, we present a quick introduction to relevant details of the \(L^2\)-metric, which are well documented in the literature, for details, the reader can consult the reference~\cite[Sec.~7.10,~7.14]{mantonTopologicalSolitons2004a}. 
Assume the vortex position describes a curve \(X(s)\) on the surface, in this way the gauge equivalence class \([(\phi, a)]\) of a solution of Bogomolny's equations describes a curve on moduli space. Up to gauge equivalence, we can assume Gauss' law holds, 
\begin{equation}
  *d\!*\!\dot{a} + \frac{i}{2}  \left( \bar{\phi} \dot{\phi} - \phi \bar{\dot{\phi}}\right) = 0,
\end{equation}
then the kinetic energy,
\begin{equation}
  T = \frac{1}{2} \int_{S} \left(|\dot{a}|^2 + |\dot{\phi}|^2\right) \operatorname{Vol},
\end{equation}
determines a K\"ahler metric on moduli space. On the plane, Samols studied the properties of this metric~\cite{samolsMathematicalPhysicsVortex1992a}, whereas for vortices with Dirichlet boundary conditions, Nasir gave a detailed description of the metric properties of moduli space~\cite{nasir_study_1998}. Let \(\chi: S\setminus{\{X\}} \to \mathbb{R}\) be  the phase of the vortex field \(\phi \), such that \(\phi = \exp(h/2 + i\chi)\), define the field 
\begin{equation}
  \eta = \frac{1}{2} \dot{h} + i \dot{\chi},
\end{equation}
such that \(\dot{\phi} = \phi \eta \). If \(D\) is an \(\epsilon \)-disk centred at \(X\), the kinetic energy can be computed as 
\begin{equation}\label{eq:energy-integral-border}
  T = i \int_{\partial S} \bar{\eta} \bar{\partial} \eta +  i \lim_{\epsilon \to 0} \int_{\partial D} \bar{\eta} \bar{\partial} \eta,
\end{equation} 
where \(\partial S\) and \(\partial D\) are oriented by the outward pointing normal. In any local trivialisation \(\varphi: U \to \mathbb{C}\) about \(X\), with coordinate \(Z = \varphi(X)\), it is well known that 
\begin{equation}
  \eta = \dot{Z} \partial_{Z} h(x; \varphi(X)).
\end{equation}
Notice that if \(h\) satisfies Dirichlet boundary conditions, the first term in Equation~\eqref{eq:energy-integral-border} vanishes, however, for Neumann boundary conditions, this is not necessarily the case, in fact, since \(*d \eta = 0\) at the boundary, \(\bar{\partial} \eta|_{\partial S} = d\eta|_{\partial_{S}}\), whereas for the second term, Samols deduction still holds, i.e., if we define the coefficient function \(b(Z)\) as,
\begin{equation}
  b(Z) = \left.\frac{\partial}{\partial \bar{z}}\right|_{z = Z} \left(h(\varphi^{-1}(z); \varphi^{-1}(Z)) - \log |z - Z|^2\right),
\end{equation}
and if under the chart \(\varphi \) the metric is \(\Omega(Z) |dZ|^2\), then 
\begin{equation}
  i \lim_{\epsilon \to 0}\int_{\partial D} \bar{\eta} \bar{\partial} \eta = \frac{1}{2} \pi \left( \Omega(Z) +  2 \frac{\partial b}{\partial Z}\right) |\dot{Z}|^2,
\end{equation}
on the other hand, for the first term in Equation~\eqref{eq:energy-integral-border},
\begin{align}
  i\int_{\partial S} \bar{\eta} \bar{\partial} \eta &= 
  i |\dot{Z}|^2 \int_{\partial S} \partial_{\bar{Z}}h d \partial_{Z}h \nonumber \\ 
  &= \frac{1}{4}|\dot{Z}|^2 \int_{\partial S} (\partial_{X}h d\partial_{Y}h - \partial_{Y}h d\partial_{X}h) \nonumber \\
  &= \frac{1}{2} |\dot{Z}|^2 \int_{\partial S} \partial_{X}h d\partial_{Y}h.
\end{align}
Therefore the kinetic energy of a moving vortex on moduli space is,
\begin{equation}\label{eq:kin-energy-1-vortex}
  T = 
  \frac{1}{2}\left(\int_{\partial S} \partial_{X}h d \partial_{Y} h\right) |\dot{Z}|^2 + \pi \left( \Omega(Z) +  2 \frac{\partial b}{\partial Z}\right) |\dot{Z}|^2.
\end{equation}
Equation~\eqref{eq:kin-energy-1-vortex} provides a formula for the \(L^2\) metric on moduli space, if the first integral cancels, then the metric depends purely on local data,  
we aim to prove that the first term is non-zero when \(S\) is rotationally symmetric with respect to an interior point. In the following, we assume \(S\) is isometric to a disk with polar coordinates \((r, \theta) \), and that the conformal factor is a function \(\Omega(r)\) of the distance to the origin. 

\begin{lemma}\label{lem:int-delx-d-dely}
  If \(h(z; Z)\) is the solution of Taubes equation with Neumann boundary conditions on the disk \(\mathbb{D}_{R} \subset \mathbb{C}\) with radial conformal factor \(\Omega(r)\) corresponding to a single vortex located at \(Z\), then    
  \begin{align}
    \int_{\partial \mathbb{D}_{R}(0)} \partial_{X}h d \partial_{Y}h = 
    \pi {\left( \partial_{X}h(R; 0)\right) }^2,
  \end{align}
  where \(Z = X + iY\). 
\end{lemma}
\begin{proof}
  Since rotations respect the origin are isometries, for any rotation  of the disk, \(z \mapsto e^{i\theta}z\), \(h\) satisfies the identity \(h(z; Z) = h(e^{i\theta}z; e^{i\theta}Z)\), whence, 
  \begin{align}
    \partial_{X}h(Re^{i\theta}; 0) &= \left.\frac{d}{dt}\right|_{t=0} h(Re^{i\theta}; t) \nonumber \\
     &= \left.\frac{d}{dt}\right|_{t=0} h(R; e^{-i\theta} t) \nonumber \\ 
     &= \cos(\theta) \partial_{X}h(R; 0) - \sin(\theta) \partial_{Y}h(R; 0).\label{eq:cos-delx-h}
  \end{align}
  Similarly, we find the identity,
  \begin{equation}
    \partial_{Y} h(Re^{i\theta}; 0) = \sin(\theta) \partial_{X} h(R; 0) + \cos(\theta) \partial_{Y}h(R; 0). \label{eq:sin-delx-h}
  \end{equation} 
  The reflection \(z \mapsto \bar{z}\) is also an isometry, hence, for any \(t \in \mathbb{R}\), \(h(R;it) = h(R; -it)\), hence \(\partial_{Y}h(R; 0) = 0\). By equations~\eqref{eq:cos-delx-h} and~\eqref{eq:sin-delx-h}, 
  \begin{align}
    \partial_{X}h(Re^{i\theta}; 0) d \partial_{Y}h(Re^{i\theta}; 0) &= 
    \cos(\theta)\partial_{X}h(R; 0) d \left(
        \sin(\theta) \partial_{X}h(R;0) 
      \right) \nonumber \\
      &= {\left( \partial_{X}h(R; 0)\right) }^2 \cos^2(\theta) d\theta.\label{eq:delx-d-dely}
  \end{align}
  The Lemma follows integrating Equation~\eqref{eq:delx-d-dely}.
\end{proof}

In the next lemma, we prove \(\partial_{X}h(R; 0) \neq 0 \), by virtue of Lemma~\ref{lem:int-delx-d-dely}, the boundary integral in Equation~\eqref{eq:kin-energy-1-vortex} for a moving vortex will be non-zero if the vortex passes through  the origin, hence, there is no localization formula in this case, because the kinetic energy depends on the boundary of the surface. Let us denote \(\partial_{X}h(z; 0)\) simply as \(\partial_{X}h\), by Taubes equation, \(\partial_{X}h\) is the solution of the elliptic equation, 
\begin{align}
  \nabla^2 \partial_{X}h = e^{h} \partial_{X}h,
\end{align}
defined on \(\mathbb{D}_{R}(0) \setminus \{0\} \) and 
subject to Neumann boundary conditions on \(\partial \mathbb{D}_{R}(0)\) and to the condition 
\begin{equation}\label{eq:partialx-h-cos-r}
  \partial_{X}h(z) = \frac{-2 \cos(\theta)}{r} + \partial_{X}\tilde{h}(z),
\end{equation} 
where \(z = re^{i\theta} \) and \(\tilde{h} \) is the solution of the regular Taubes equation~\eqref{eq:taubes-reg}. Let \(u = \partial_{X}\tilde{h} \), and let us denote the Euclidean Laplacian as \({\nabla}_{e}^2 = \Omega \nabla^2 \). Notice that at vortex position \(Z = 0\), \(U(1)\)-invariance implies \(h\) is a function of \(r = |z|\), then \(u\) is the solution on the disk \(\mathbb{D}_{R}(0) \) of the PDE 
\begin{equation}\label{eq:nabla2-u}
  {\nabla}_{e}^2 u = f(r) u + g, 
\end{equation} 
with Neumann boundary conditions, where 
\begin{align}
  f(r) = \Omega(r)e^{h} \qquad \text{and} \qquad g = -\frac{2}{r} \cos(\theta) \Omega(r) e^{h}.
\end{align}
Since \(e^{h}\) has a zero of order 2 at \(r = 0\), both functions extend smoothly to \(r = 0\). Equation~\eqref{eq:nabla2-u} implies the existence of a function \(a(r)\) at least of class \(C^1\), such that \(u = a(r)\cos(\theta) \), we use this fact in the following Lemma. 

\begin{lemma}\label{lem:delx-h-r}
  Let \(h \) be the solution of the Taubes equation on the disk \(\mathbb{D}_{R}(0) \) with Neumann boundary conditions and vortex position at the origin, then \(\partial_{X}h(R) \neq 0 \). 
\end{lemma}
\begin{proof}
  By Equation~\eqref{eq:partialx-h-cos-r}, 
  \begin{equation}\label{eq:partialx-h-polar-coords}
    \partial_{X}h = 
    \frac{-2\cos(\theta)}{r} + u = 
    \left(-\frac{2}{r} + a(r)\right)\cos(\theta),
  \end{equation}
  for some function \(a(r)\). Let \(\epsilon > 0 \) be a small number such that \(-2 \epsilon^{-1} + a(\epsilon) < 0\).  Assume towards a contradiction \(\partial_{X}h(R) = 0\), let \(L = \nabla^2_{e} - \Omega \, e^{h}\), then \(L \) is an elliptic, negative operator, such that \(L\partial_{X}h = 0\), we can apply the strong maximum principle in the half annulus \(\epsilon \leq r \leq R \), \(-\pi/2 \leq \theta \leq \pi/2 \), to deduce by Equation~\eqref{eq:partialx-h-polar-coords}  that on this region, 0 is the global maximum of \(\partial_{X}h\), hence, \(\partial_{X}h(R) > \partial_{X}h(z) \) for \(z\) in a small disk \(D\) contained in the interior of the half-annulus and such that \(R \in \partial D\). By Hopf's lemma~\cite[Lem.~3.4]{gilbargEllipticPartialDifferential2015}, the normal derivative \(\partial_{n}\!\left(\partial_{X}h\right)\) with respect to the Euclidean metric satisfies 
  \begin{align}
    \partial_{n}\!(\partial_{X}h)(R) > 0,
  \end{align}
  this implies the normal derivative with respect to the metric we consider on \(\mathbb{D}_{R}(0)\) is also positive, since our metric is conformally flat. This contradicts the Neumann boundary condition. 
\end{proof}
Recall Equation~\eqref{eq:kin-energy-1-vortex} for the kinetic energy of a vortex moving on moduli space, lemmas~\ref{lem:int-delx-d-dely} and~\ref{lem:delx-h-r} prove that at the origin, the energy also depends on the value of \(\partial_X h \) at the boundary of the disk, whereas for vortices on surfaces with no boundary and vortices with Dirichlet boundary conditions, the energy, and hence the \(L^2\) metric, only depends on vortex position and leads to the localization formula.

\section{Conclusion and outlook}\label{sec:conclusion-outlook}

In this work we proved that on any surface with boundary there exists a
unique solution of the Taubes equation with Neumann boundary
conditions, provided the location of the core set of the
Higgs field and that the Bradlow 
bound~\eqref{eq:bradlow-cond} holds. We extended previous
results developed for vortices on the quotient of closed surfaces by a
reflection symmetry. We also proved that the Bradlow bound
found by Manton and Zhao is necessary
and sufficient for the existence of vortices satisfying Neumann 
boundary conditions on the surface. As a consequence, the moduli space of vortices with Neumann boundary conditions is well defined. For
closed surfaces it is known that this space is a stratified, complete K\"ahler
manifold, with a metric depending on local data, regardless the surface shape. We found evidence that for Neumann boundary conditions the \(L^2 \)
metric is no longer a local object, and that the metric depends on boundary data of the vortices, opening the possibility that this is a general phenomena, however, a more detailed study of the geometry of moduli space would be necessary, either numeric or analytic. 

\section*{Acknowledgements}

René Israel García Lara acknowledges the support by the UNAM Postdoctoral Program (POSDOC).

\appendix

\section{Deforming Green functions on a surface with
  boundary}\label{sec:green}

In this section we show how to construct Green functions
satisfying Neumann conditions on
\( S \). If \(\Delta \subset \bar{S} \times \bar{S} \) is the main diagonal set,
recall~\cite{aubin_nonlinear_2013} \(G: \bar S\times \bar S\setminus \Delta
  \to \mathbb{R} \) is a Green function for the Laplace operator, with
Dirichlet boundary conditions, if \(G(P,Q) \) is a solution of the equation
\begin{align}
  \label{eq:1}
  -\nabla^2_{P} G(\cdot,Q) = \delta_{Q},
\end{align}
in the sense of distributions and whenever \(P \) or \(Q \) are boundary
points, then \(G(P,Q) = 0 \). For any test function \(\varphi \in  C^2(\bar S) \), \(G \) satisfies the integral equation
\begin{align}
  \label{eq:green-integral-identity}
  \varphi(Q) = -\int_{S}G(P,Q)\nabla^2\varphi(P)\,dV(P) -
  \int_{\partial S} \partial_{nP}\,G(P,Q)\varphi(P)\,ds(P).
\end{align}

\begin{lemma}\label{lem:green-interior}
  For any \( Q \in S \), there exists a function \( G_Q \), smooth on \( S\setminus \{Q\} \) and continuous in \( \bar S\setminus
    \{Q\} \), such that,
  \begin{align}
    -\nabla^2G_Q   & = \delta_Q - \frac{1}{A}, \\
    \partial_n G_Q & = 0,
  \end{align}
  where \( A \) is the area of \( S \).
\end{lemma}
\begin{proof}
  We choose a tubular neighbourhood of \( \partial S \) such that \( Q \) is
  in the exterior. By choosing a bump function \( \varphi: \bar S \to
    \mathbb{R} \) such that \( \varphi \equiv 1 \) in the tubular
  neighbourhood and \( \varphi \equiv 0 \) in a neighbourhood of \( Q \), we
  can defined a new smooth function \( H(P) = \varphi(P)\,G(P, Q) \) such
  that \( \partial_{n}H = \partial_{nP}G(P, Q) \),
  by equation~\eqref{eq:green-integral-identity}
  \begin{align*}
    \int_{\partial S} \partial_{n} H(P)\,ds(P) = -1.
  \end{align*}
  Let \( f \) be a solution of
  \begin{align}
    \label{eq:2}
    -\nabla^2 f  & = -\nabla^2 H - \frac{1}{A}, \\
    \partial_n f & = 0,
  \end{align}
  then we can define \( G_Q = G(\cdot, Q) -H + f \).
\end{proof}

\begin{lemma}\label{lem:green-boundary}
  If \( Q \in \partial S \), there exists a solution \(H_Q \) of the equation
  \begin{align}
    \label{eq:3}
    \nabla^2 H_Q   & = \frac{1}{A}, \\
    \label{eq:4}
    \partial_n H_Q & = \delta^1_{Q}.
  \end{align}
\end{lemma}
\begin{proof}
  We follow~\cite[pg.~107]{aubin_nonlinear_2013} adjusting for the
  fact that \( Q \)  is a boundary point. Firstly, we show the existence
  of a function \( G_{Q} \) such that \( \partial_{n}G_Q = \delta^1_{Q} \).
  Let \( (U,\varphi) \) be a chart
  around \( Q \) such that \( \varphi \) maps \( U \) onto the upper half plane and \( \varphi(Q) = 0 \). We assume that in this chart the metric takes the
  form
  \begin{align*}
    \Omega \, (dx_1^2 + dx_2^2).
  \end{align*}
  Let \((x_1, x_2) \) be the coordinates of an arbitrary point \(P\in U \) and let \(\delta_1 < \delta_2 \) be two small positive numbers, we define the
  function \(r(P) = {(x_1^2 + x_2^2)}^{1/2} \), we can choose a  positive decreasing function \(f(r) \) which is
  1 for \(r(P) \leq \delta_1 \) and 0 for \(r(P) \geq \delta_2 \). By means of the coordinate chart, we define the function \(G_Q\in C^\infty(\bar S) \)
  such that \(G_Q(P) = -\pi^{-1}f(r)\log(r) \) for \(P \in U \) and \( G_Q(P)
    = 0 \) otherwise. Let \( B_Q(\epsilon) = \{P \in U \mid r(P) \leq
    \epsilon \} \) and let \( \psi \in C^2(\bar S) \) be a test function,  we
  aim to compute
  \begin{align*}
    \int_{\partial S}\partial_{n}
    G_Q\,\psi \, ds  = \lim_{\epsilon \to 0}\int_{\partial \left(\bar S
      \setminus B_{ Q } (\epsilon)\right)}\partial_{n} G_Q\,\psi \, ds.
  \end{align*}
  \( G_Q \) is 0 for points in the exterior of \( U \), whereas for points in
  \( U\cap \partial S \), the normal exterior vector can be parametrised
  as a multiple of \( \partial_2 \) and for points in \( \partial
    B_Q(\epsilon) \) it points in the opposite direction of the radial
  vector field \( x_1\partial_1 + x_2\partial_2 \). Whence,
  \begin{align*}
    \int_{\partial S}\partial_{n}
    G_Q\,\psi \, ds & = \lim_{\epsilon\to 0} \int_{\{(x_1,x_2)\mid x_2=0,\,
      x_1^2\geq
      \epsilon \}} \partial_{n} G_{Q} \,\psi \, ds +
    \lim_{\epsilon\to 0
    } \int_{\partial B_Q(\epsilon)}
    \partial_{n}G_Q\,\psi \, ds
    \\
                    & = \lim_{\epsilon\to
    } \int_{\partial B_Q(\epsilon)}
    \partial_{n} G_Q\,\psi \, ds,
    \\
                    & = \psi(Q).
  \end{align*}
  where in the first equation, the first integral cancels because the
  normal exterior derivative is a multiple of \( \partial_2 \) and
  \( \partial_2 r|_{x_2=0} = 0 \). Therefore, \( G_Q \) is the required
  function. Secondly, we fix a solution \( u \) of the problem
  \begin{align*}
    \nabla^2 u  & =  -\nabla^2G_Q +  \frac{1}{A}, \\
    \partial_n u & = 0.
  \end{align*}
  Defining \( H_Q = G_Q + u \) we conclude the statement of the lemma.
\end{proof}

\bibliography{references}
\bibliographystyle{plain}

\end{document}